\documentclass[journal]{IEEEtran}
\usepackage{stfloats}
\usepackage{makecell}
\usepackage{graphicx}
\usepackage[cmex10]{amsmath}
\usepackage{cases}
\usepackage[tight,footnotesize]{subfigure}
\usepackage{amsthm} 
\usepackage{cite}
\usepackage{citesort}
\usepackage{amssymb}
\allowdisplaybreaks[4]
\usepackage{algorithm}
\usepackage{algorithmic}
\usepackage{multirow}
\usepackage{amsmath}
\usepackage{xcolor}
\usepackage{CJK}
\usepackage{subeqnarray}
\usepackage{cases}
\usepackage{enumerate}

\usepackage{stfloats}

\newtheorem{theorem}{\hskip\parindent\bf{Theorem}}

\ifCLASSINFOpdf
\else
\fi

\begin{document}

\title{\LARGE{Task and Bandwidth Allocation for UAV-Assisted Mobile Edge Computing with Trajectory Design}}

\author{\normalsize{Xiaoyan Hu$^{\dag}$, Kai-Kit Wong$^{\dag}$, Kun Yang$^{\ddag*}$ and Zhongbin Zheng$^{\S}$} \\
$^{\dag}$Department of Electronic and Electrical Engineering, University College London, London, UK\\
$^{\ddag}$School of Computer Science and Electronic Engineering, University of Essex, Colchester, UK \\
$^{*}$School of Communication and Information Engineering, UESTC, Chengdu, P. R. China\\
$^{\S}$East China Institute of Telecommunications, Shanghai, P. R. China\\
Email:$^{\dag}$\{xiaoyan.hu.16, kai-kit.wong\}@ucl.ac.uk, $^{\ddag*}$kunyang@\{essex.ac.uk, uestc.edu.cn\},$^{\S}$ben@ecit.org.cn
}
\maketitle
\vspace{-2cm}
\begin{abstract} 
In this paper, we investigate a mobile edge computing (MEC) architecture with the assistance of an unmanned aerial vehicle (UAV). The UAV acts as a computing server to help the user equipment (UEs) compute their tasks as well as  a relay to further offload the UEs' tasks to the access point (AP) for computing. The total energy consumption of the UAV and UEs is minimized by jointly optimizing the task allocation, the bandwidth allocation and the UAV's trajectory, subject to
the task constraints, the information-causality constraints, the bandwidth allocation constraints, and the UAV's trajectory constraints. The formulated optimization problem is nonconvex, and we propose an alternating algorithm to optimize the parameters iteratively. The effectiveness of the algorithm is verified by the simulation results, where great performance gain is achieved in comparison with some practical baselines, especially in handling the computation-intensive and latency-critical tasks.
\end{abstract}

\IEEEpeerreviewmaketitle

\vspace{-4mm}
\section{Introduction}\label{sec:Introduction}
With the prevail of smart devices, various mobile applications spring up which demand  more and more computing resources and will definitely increase the burden of the user equipment (UEs). 
Mobile edge computing (MEC) aiming at liberating the resource-limited UEs from heavy computing workload through computation offloading, has been regarded as a key technology of future intelligent wireless networks \cite{S_P.Mach17Mobile,S_Y.Mao17ASurve}.
Recently, MEC has been widely used in cellular networks, focusing on improving the energy efficiency or reducing the latency of various cellular-based MEC systems \cite{S.Sardellitti.15Joint,J_C.You17Energy,T.Q.Dinh17Offloading,J_X.Hu18Edge,J_H.Sun19Joint,J_X.Hu18Wireless}. 

In order to further improve the system performance of edge computing, the technology of the unmanned aerial vehicle (UAV) has been leveraged in  MEC architectures, since the UAV communication has many attractive advantages, such as easy deployment, flexible movement, and line-of-sight (LoS) connections, etc.,  \cite{J_zeng16Wireless,J_M.M.Azari18Ultra,J_Xu18Uav}. 
The performance improvement of the UAV-enabled MEC architecture has been shown to be substantial \cite{J_Jeong18Mobile,J_F.Zhou18Computation}. 
 A UAV-based MEC system is investigated in \cite{J_Jeong18Mobile}, where a moving UAV helps UEs compute their offloaded tasks and the total mobile energy consumption is minimized by optimizing the task-bit allocation and the UAV's trajectory. Later in \cite{J_F.Zhou18Computation}, a wireless powered UAV-enabled MEC system is studied, where the UAV provides energy as well as MEC services for the UEs. The computation rate maximization problems are addressed by alternating algorithms under both the partial and binary computation offloading modes. 


The aforementioned MEC works concentrate either on the cellular-based MEC networks or the UAV-enabled MEC architectures by exploiting the computing capability solely at the AP or the UAV. However, it is impossible to take full use of the computing resources at the AP  if the UEs' links to the AP are seriously degraded. Besides, it is risky to rely only on the UAVs for completing UEs' computation-intensive latency-critical tasks since the UAVs are also resource-limited \cite{X.Hu18UAV}. For these reasons, this paper studies a UAV-assisted MEC architecture, where the computing resources at the UAV and the AP are utilized at the same time. In addition, the energy-efficient LoS transmissions of the UAV have been fully exploited since the UAV is not only served as a mobile computing server to help the UEs compute their tasks  but also as a relay to further offload UEs' tasks to the AP for computing.


\vspace{-2mm}
\section{System Model and Problem Formulation}\label{sec:system}
\subsection{System Model}\label{sec:Channel Model}
A UAV-assisted  MEC system is considered in this paper, which consists of an AP, a cellular-connected UAV, and $K$ ground UEs, all being equipped with a single antenna. The UAV and UEs are all assumed to have limited computing resources  powered by their embedded battery, while the AP  is endowed with an ultra-high performance processing server.
A three-dimensional (3D) Euclidean coordinate system is adopted, whose coordinates are measured in meters. The locations of the AP and all the UEs are fixed on the ground with zero altitude, denoted as $\mathbf{\widetilde{v}}_0=(x_0,y_0,0)$ and $\mathbf{\widetilde{v}}_k=(x_k,y_k,0)$ for UE $k\in\mathcal{K}=\{1,\dots,K\}$.  
We assume that the UAV flies at a fixed altitude $H>0$ during the task completion time $T$, which is the minimum altitude appropriate to the work terrain. 

For ease of exposition, the finite task completion time $T$ is discretized into $N$ equal time slots each with a duration of $\tau=T/N$,
where $\tau$ is sufficiently small such that the UAV's location can be assumed to be unchanged during each slot. 
The initial and final horizontal locations of the UAV are preset as $\mathbf{u}_\mathrm{I}=(x_\mathrm{I}, y_\mathrm{I})$ and $\mathbf{u}_\mathrm{F}=(x_\mathrm{F}, y_\mathrm{F})$. Let $\mathcal{N}=\{1,\dots,N\}$ denote the set of the $N$ time slots. At the $n$-th time slot, the UAV's horizontal location is denoted as $\mathbf{u}[n]\equiv\mathbf{u}(n\tau)=(x[n],y[n])$ with $\mathbf{u}[0]=\mathbf{u}_\mathrm{I}$ and $\mathbf{u}[N]=\mathbf{u}_\mathrm{F}$. It is assumed that the UAV flies with a constant speed in each time slot, denoted as $\mathbf{v}[n]$, which should satisfy $\|\mathbf{v}[n]\|=\frac{\|\mathbf{u}[n]-\mathbf{u}[n-1]\|}{\tau}\leq V_\mathrm{max}$,
with a predetermined maximum speed $V_\mathrm{max}\geq\|\mathbf{u}_\mathrm{F}-\mathbf{u}_\mathrm{I}\|/T$.

Similar to \cite{J_Zeng17Energy}, 
the wireless channels between the UAV and the AP as well as the UEs are assumed to be dominated by LoS links. Thus, the channel power gain between the UAV and the AP and between the UAV and  UE $k$ at the time slot $n$ can be, respectively, given by
\begin{align}
&h_{\mathrm{AP}}[n]=\frac{h_0}{\|\mathbf{u}[n]-\mathbf{v}_0\|^2+H^2}, \ n\in \mathcal{N}, \label{eq:channel_AP} \\
&h_k[n]=\frac{h_0}{\|\mathbf{u}[n]-\mathbf{v}_k\|^2+H^2}, \ k\in\mathcal{K}, n\in \mathcal{N}, \label{eq:channel_k}
\end{align}
where $h_0$ is the channel power gain at a reference distance of one meter; $\mathbf{v}_0=(x_0,y_0)$ and $\mathbf{v}_k=(x_k,y_k)$ denote the horizontal locations of the AP and UE $k\in\mathcal{K}$, respectively.
In this paper, the direct links between UEs and the AP are assumed to be negligible due to e.g., severe blockage, which means that the UEs cannot directly offload their task-input bits to the AP unless with the assistance of the UAV.

\vspace{-3mm}
\subsection{Computation Task Model and Energy Consumption}\label{sec:Task Model}
It is assumed that each UE has a  computation-intensive task, which is denoted as a positive tuple $[I_k, C_k, T_k]$. Here, $I_k$ denotes the size (in bits) of the computation task-input data which is bit-wise independent and can be arbitrarily divided, $C_k$ is the amount of required computing resource for computing 1-bit of input data,  and $T_k$ is the maximum tolerable latency with $T_k\leq T, k\in \mathcal{K}$.\footnote{ In our considered scenario, we assume that the output data sizes of the computation tasks are quite small that can be ignored especially compared with the input data sizes of the computation-intensive tasks.}
In this paper, we only consider the case that $T_k=T$ for all $k\in\mathcal{K}$.
The UAV acts as an assistant to  help the UEs complete their computation tasks by providing both  MEC and relaying services. For the MEC service, the UAV shares its computing resources with the UEs to help compute their tasks; while for the relaying service, the UAV forwards part of the UEs' offloaded tasks to the AP for computing.
Hence, the UEs can accomplish their computation tasks in a partial offloading fashion \cite{S_Y.Mao17ASurve} with the following three ways. 
\subsubsection{Local Computing at UEs}
Each UE can perform local computing and computation offloading simultaneously. 
To efficiently use the energy for local computing, the UEs leverage a dynamic voltage and frequency scaling (DVFS) technique to adaptively  adjust their CPU frequency in each slot \cite{J_Zhang13Energy}. 
The CPU frequency of  UE $k$ during time slot $n$ is denoted as $f_k[n]$ (cycles/second) for computing $l{_k^{\mathrm{loc}}}[n]=\tau f_k[n]/C_k$  task-input bits. Thus, the energy consumption of UE $k$ during time slot $n$ for local computing can be expressed as
\begin{align}
&E{_k^{\mathrm{loc}}}[n]=\tau\kappa_k f_k^3[n]\equiv\frac{\kappa_kC_k^3}{\tau^2} (l{_k^{\mathrm{loc}}}[n])^3, \ k\in\mathcal{K}, n\in \mathcal{N}, \label{eq:Local_energy_k}
\end{align}
where $\kappa_k$ is the effective capacitance coefficient of UE $k$ that depends on  its processor's chip architecture. 

\subsubsection{Task Offloaded to the UAV for Computing}
The UEs' remaining task-input data should be computed remotely, first by offloading to the UAV, and then one part of the data being computed at the UAV while the other part further offloaded to the AP for computing. In order to avoid interference among the UEs during the offloading process, we adopt the time-division multiple access (TDMA) protocol.
Each slot is further divided into $K$ equal durations $\delta=T/(NK)$, and UE $k$ offloads its task-input data in the $k$-th duration. Let $l_k^{\mathrm{off}}[n]$ denote the offloaded bits of UE $k$ in its allocated duration at time slot $n$, and thus  the corresponding energy consumption of UE $k$ at slot $n$ for computation offloading can be calculated  as
\begin{align}\label{eq:E_off_k}
E{_k^{\mathrm{off}}}[n]&=\delta p_k^{\mathrm{off}}[n]
\equiv\frac{\delta N_0}{h_k[n]}\bigg(2^{\frac{l_k^{\mathrm{off}}[n]}{\delta B{_k^\mathrm{off}}[n]}}-1\bigg), 
\end{align}
where $p_k^{\mathrm{off}}[n]$ and $B{_k^\mathrm{off}}[n]$ are
the related transmit power and bandwidth,
and $N_0$ denotes the noise power at the UAV.\footnote{Without loss of generality, we assume that the noise power at any node in the system is considered the same as $N_0$.} 

Assume that the UAV also adopts the DVFS technique to improve its energy efficiency for computing, and its adjustable CPU frequency in the $k$-th duration of slot $n$  is denoted as $f_{\mathrm{U},k}[n]$ for computing UE $k$'s  $l_{\mathrm{U},k}[n]=\delta f_{\mathrm{U},k}[n]/C_k$ offloaded task-input bits. Hence, the energy consumption of the UAV for computing UE $k$'s task at slot $n$ can be  given as 
\begin{align}
E_{\mathrm{U},k}[n]&=\delta\kappa_{\mathrm{U}} f_{\mathrm{U},k}^3[n]\equiv\frac{\kappa_{\mathrm{U}}C_k^3 }{\delta^2} l_{\mathrm{U},k}^3[n], \ k\in\mathcal{K}, n\in \mathcal{N}, \label{eq:UAV_energy_k}
\end{align}
where $\kappa_\mathrm{U}$ is the effective capacitance coefficient of the UAV.

\subsubsection{Task Offloaded to the AP for Computing}
Part of the UEs' offloaded task-input data at the UAV will be offloaded to the AP's processing server for computing. To better distinguish the offloading signals from different UEs, the TDMA protocol with $K$ equal time division ($\delta=T/(NK)$) is also adopted in this case. Let $l{_{\mathrm{U},k}^{\mathrm{off}}}[n]$ denote the number of UE $k$'s task-input bits being offloaded from the UAV to the AP at time slot $n$. Thus, the corresponding energy consumption of the  UAV for offloading UE $k$'s task at slot $n$ can be calculated as
\begin{align}\label{eq:E_UAV_off_k}
E{_{\mathrm{U},k}^{\mathrm{off}}}[n]&=\delta p{_{\mathrm{U},k}^{\mathrm{off}}}[n]
\equiv\frac{\delta N_0}{h_{\mathrm{AP}}[n]}\bigg(2^{\frac{l{_{\mathrm{U},k}^{\mathrm{off}}}[n]}{\delta B{_{\mathrm{U},k}^\mathrm{off}}[n]}}-1\bigg),
\end{align}
where $p{_{\mathrm{U},k}^{\mathrm{off}}}[n]$ and $B{_{\mathrm{U},k}^\mathrm{off}}[n]$ are the corresponding transmit power  and  the allocated bandwidth of the UAV.
As the AP is integrated with an ultra-high-performance processing server, the computing time is negligible. It is assumed that the UAV is equipped with a data buffer with sufficiently large size, and it is capable of storing each UE's offloaded data. 

In fact, the energy consumption for UAV's propulsion is also considerable which is greatly affected by the UAV's trajectory, and hence should be taken into account. 
With the assumption that the time slot duration $\tau$ is  sufficiently small, the UAV's flying during each slot can be regarded as  straight-and-level flight with constant speed $\mathbf{v}[n]$.
Taking a fixed-wing UAV as an example \cite{J_Zeng17Energy,J_Y.Zeng19Accessing}, its propulsion energy consumption
at time slot $n$ can be expressed as
\begin{align}\label{eq:E_fly}
E{_\mathrm{U}^{\mathrm{fly}}}[n]=\tau\bigg(\theta_1\|\mathbf{v}[n]\|^3+\frac{\theta_2}{\|\mathbf{v}[n]\|}\bigg), \ n\in \mathcal{N},
\end{align}
where $\theta_1$ and $\theta_2$ are two parameters related to the UAV's weight, wing area, wing span efficiency, and air density, etc.
Combining with the above analysis, we obtain the total energy consumption of UE $k$ and the UAV in each time slot $n$ as
\begin{align}
\hspace{-2.5mm}E_{k}[n]&=E{_k^{\mathrm{loc}}}[n]+E{_k^{\mathrm{off}}}[n], \ k\in\mathcal{K}, n\in \mathcal{N},\label{eq:E_k_n}\\
\hspace{-2.5mm}E_{\mathrm{U}}[n]&=\sum\limits_{k{\rm{ = }}1}^K \Big(E_{\mathrm{U},k}[n]+E{_{\mathrm{U},k}^{\mathrm{off}}}[n]
\Big)+E{_\mathrm{U}^{\mathrm{fly}}}[n],\ n\in \mathcal{N}.\label{eq:E_UAV_k_N}
\end{align}

\vspace{-6mm}

\subsection{Problem Formulation}\label{sec:problem}
Considering that the battery-based UEs and  UAVs are usually power-limited, one major problem the UAV-assisted MEC system faces is energy. Hence, in this paper, we try to minimize the total energy consumption (TEC) of the UAV and all the UEs during the task completion time $T$. In the previous subsection, we have obtained the energy consumption of the UEs and the UAV for task offloading and computation.
In our considered scenario, the UEs' task allocation parameters in $\mathbf{l}\triangleq\{l_{k}^{\mathrm{loc}}[n],l{_{k}^{\mathrm{off}}}[n], l_{\mathrm{U},k}[n], l{_{\mathrm{U},k}^{\mathrm{off}}}[n]\}_{k\in\mathcal{K},n\in\mathcal{N}}$, the bandwidth allocation parameters in $\mathbf{B}\triangleq\{B{_k^\mathrm{off}}[n],$ $B{_{\mathrm{U},k}^{\mathrm{off}}}[n]\}_{k\in\mathcal{K},n\in\mathcal{N}}$ as well as the UAV's trajectory parameters in $\mathbf{u}\triangleq\{\mathbf{u}[n]\}_{n\in\mathcal{N}}$ will be optimized to minimize the TEC.
To this end, the TEC minimization problem is formulated as problem (P1) below
\begin{subeqnarray}\label{eq:WSECM1} 
&&\hspace{-7mm} \underset{\mathbf{l},\mathbf{B},\mathbf{u}}{\min}\
 \sum\limits_{n{\rm{ = }}1}^N\bigg(E_{\mathrm{U}}[n]+\sum\limits_{k{\rm{ = }}1}^KE_k[n]\bigg)\slabel{eq:WSECM1_0}\\
&&\hspace{-7mm}\mathrm{s.t.}\
 \sum\limits_{i{\rm{ = }}2}^n \left(l_{\mathrm{U},k}[i]+l{_{\mathrm{U},k}^{\mathrm{off}}}[i]\right)\leq \sum\limits_{i{\rm{ = }}1}^{n-1} l{_{k}^{\mathrm{off}}}[i],n\in\mathcal{N}_2, k\in\mathcal{K}, \quad~~ \slabel{eq:WSECM1_1}\\
&&\hspace{-7mm} ~~~~\sum\limits_{n{\rm{ = }}2}^{N} \left(l_{\mathrm{U},k}[n]+l{_{\mathrm{U},k}^{\mathrm{off}}}[n]\right)= \sum\limits_{n{\rm{ = }}1}^{N-1} l{_{k}^{\mathrm{off}}}[n], \ k\in\mathcal{K}, \slabel{eq:WSECM1_4}\\
&&\hspace{-7mm} ~~~~\sum\limits_{n{\rm{ = }}1}^{N}l{_{k}^{\mathrm{loc}}}[n]+\sum\limits_{n{\rm{ = }}1}^{N-1} l{_{k}^{\mathrm{off}}}[n]=I_k,\ k\in\mathcal{K}, \slabel{eq:WSECM1_3}\\
&&\hspace{-7mm}~~~~ B{_k^\mathrm{off}}[n]+B{_{\mathrm{U},k}^{\mathrm{off}}}[n]=B, \ n\in\mathcal{N}, \ k\in\mathcal{K}, \slabel{eq:WSECM1_12}\\
&&\hspace{-7mm}~~~~ l_k^{\mathrm{loc}}[n]\geq0, \ n\in\mathcal{N}, \ k\in\mathcal{K}, \slabel{eq:WSECM1_6}\\
&&\hspace{-7mm} ~~~~l_k^{\mathrm{off}}[N]=0, \ l_k^{\mathrm{off}}[n]\geq0, \ n\in\mathcal{N}_1,\  k\in\mathcal{K}, \slabel{eq:WSECM1_11}\\
&&\hspace{-7mm} ~~~~l_{\mathrm{U},k}[1]=0,\ l_{\mathrm{U},k}[n]\geq0, \ n\in\mathcal{N}_2,\  k\in\mathcal{K}, \slabel{eq:WSECM1_14}\\ 
&&\hspace{-7mm}~~~~ l{_{\mathrm{U},k}^{\mathrm{off}}}[1]=0, \ l{_{\mathrm{U},k}^{\mathrm{off}}}[n]\geq0, \ n\in\mathcal{N}_2,\  k\in\mathcal{K}, \slabel{eq:WSECM1_7}\\ 
&&\hspace{-7mm}~~~~ B{_k^\mathrm{off}}[N]=0, B{_k^\mathrm{off}}[n]\geq0,\ n\in\mathcal{N}_1, k\in\mathcal{K},\slabel{eq:WSECM1_13}\\
&&\hspace{-7mm} ~~~~B{_{\mathrm{U},k}^{\mathrm{off}}}[1]=0, B{_{\mathrm{U},k}^{\mathrm{off}}}[n]\geq0,
\ n\in\mathcal{N}_2, \  k\in\mathcal{K},\slabel{eq:WSECM1_15}\\ 
&&\hspace{-7mm} ~~~~\mathbf{u}[0]=\mathbf{u}_\mathrm{I}, \ \mathbf{u}[N]=\mathbf{u}_\mathrm{F}, \slabel{eq:WSECM1_9}\\
&&\hspace{-7mm}~~~~ \|\mathbf{u}[n]-\mathbf{u}[n-1]\|\leq V_\mathrm{max}\tau, \ n\in\mathcal{N},  \quad \slabel{eq:WSECM1_10}
\end{subeqnarray}
where  $\mathcal{N}_1=\{1,\dots,N-1\}$.
In (P1), \eqref{eq:WSECM1_1} is an information-causality constraints, indicating that the UAV can only compute or forward the task-input data that has already been received from the UEs with one slot processing delay. Hence, the UEs should not offload at the last slot, while the UAV should not compute or forward the received UEs' data  at the first slot.
\eqref{eq:WSECM1_4} and \eqref{eq:WSECM1_3} are the UEs' computation task constraints to make sure that all the UEs' computation task-input data have been computed. The bandwidth constraints are in \eqref{eq:WSECM1_12}.
\eqref{eq:WSECM1_9} and \eqref{eq:WSECM1_10}  specify the UAV's initial and final horizontal locations, and its maximum speed constraints.

\vspace{-1mm}
\section{Algorithm Design}\label{algorithm_design1}
The problem (P1) is a complicated non-convex optimization problem because of the non-convex objective function where non-linear couplings exist among the variables $l{_{k}^{\mathrm{off}}}[n]$ and $B{_k^\mathrm{off}}[n]$, $l{_{\mathrm{U},k}^{\mathrm{off}}}[n]$ and $B{_{\mathrm{U},k}^{\mathrm{off}}}[n]$ which are further coupled with the trajectory of the UAV, i.e., $\mathbf{u}[n]$. To address these issues, we propose a three-step alternating optimization algorithm, where the task allocation $\mathbf{l}$, the bandwidth allocation $\mathbf{B}$, and the UAV's trajectory $\mathbf{u}$ are obtained alternatively.
The details for the three steps of the algorithm are presented as follows.

\vspace{-3mm}
\subsection{Task Allocation}\label{computation_resource_scheduling}
A sub-problem of (P1) is the computation task allocation problem (P1.1), where the UAV's trajectory $\mathbf{u}$ and bandwidth allocation $\mathbf{B}$ are given as fixed. In this case, the time-dependent channels  $\{h_{\mathrm{AP}}[n]\}_{n\in\mathcal{N}}$ and $\{h_k[n]\}_{k\in\mathcal{K},n\in\mathcal{N}}$ defined in \eqref{eq:channel_AP} and \eqref{eq:channel_k} are  known with the fixed $\mathbf{u}$. Besides, the non-linear couplings among the offloading task-input \mbox{bits} ($l{_{k}^{\mathrm{off}}}[n], l{_{\mathrm{U},k}^{\mathrm{off}}}[n]$) with their corresponding allocated bandwidths ($B{_k^\mathrm{off}}[n], B{_{\mathrm{U},k}^{\mathrm{off}}}[n]$) no longer exist. The computation task allocation problem (P1.1) is convex with a convex objective function and constraints, which is expressed as
\begin{subeqnarray}\label{eq:RS1}
 ({\rm P1.1}):
 \underset{\mathbf{l}}{\min} &&\hspace{-4mm}
  \sum\limits_{n{\rm{ = }}1}^N\left(E{_{\mathrm{U}}^{(1)}}[n]+\sum\limits_{k{\rm{ = }}1}^KE_k[n]\right) \quad\quad\slabel{eq:RS1_0}\\
\mathrm{s.t.}
&&\hspace{-4mm} \eqref{eq:WSECM1_1}-\eqref{eq:WSECM1_3}, \ \eqref{eq:WSECM1_6}-\eqref{eq:WSECM1_7}, \slabel{eq:RS1_1} 
\end{subeqnarray}
where $E{_{\mathrm{U}}^{(1)}}[n]=\sum\limits_{k{\rm{ = }}1}^K \Big(E_{\mathrm{U},k}[n]+E{_{\mathrm{U},k}^{\mathrm{off}}}[n]\Big)$.
In order to gain more insights of the solution, we leverage the Lagrange method  \cite{B_Boyd04Convex} to solve problem (P1.1), and the optimal solution of problem (P1.1) is given in the following theorem.

\begin{theorem}\label{theorem_1}
The optimal solution of problem (P1.1) related to UE $k\in\mathcal{K}$ is given in \eqref{eq:f_k}--\eqref{eq:l_off_Uk} below

\vspace{-2mm}
{\small{
\begin{align}
&\hspace{-1mm}l{_{k}^{\mathrm{loc}*}}[n]=\frac{\tau}{C_k}\sqrt{\frac{[\beta{_k^*}]^+}{3C_k\kappa_k}}, \  n\in \mathcal{N}, \label{eq:f_k}\\
&\hspace{-1mm}l{_{k}^{\mathrm{off}*}}[n]=\left\{
\begin{aligned}
&\delta B{_k^{\mathrm{off}}}[n]\bigg[\log_2\rho_k^{\mathrm{off}}[n]\Big[\widehat{\lambda}_{k,n}^*+\beta{_k^*}-\eta{_k^*}\Big]^+\bigg]^+, n\in \mathcal{N}_1, \\
&0, \quad n=N, \\
\end{aligned}\right. \label{eq:l_k}\\
&\hspace{-1mm}l{_{\mathrm{U},k}^*}[n]=\left\{
\begin{aligned}
&\frac{\delta}{C_k}\sqrt{\frac{ \left[\eta{_k^*}-\widetilde{\lambda}_{k,n}^*\right]^+}{3C_k\kappa_\mathrm{U}}}, n\in \mathcal{N}_2, \\
&0, \quad n=1, \\
\end{aligned}\right. \label{eq:f_Uk}\\
&\hspace{-1mm}l{_{\mathrm{U},k}^{\mathrm{off}*}}[n]=\left\{
\begin{aligned}
&\delta B{_{\mathrm{U},k}^{\mathrm{off}}}[n]\left[\log_2\rho_{\mathrm{U},k}^{\mathrm{off}}\left[\eta{_k^*}-\widetilde{\lambda}_{k,n}^* \right]^+\right]^+, n\in \mathcal{N}_2,  \\
&0, \quad n=1, \\
\end{aligned}\right. \label{eq:l_off_Uk}
\end{align}
}}
\hspace{-1.3mm}where $\rho_k^{\mathrm{off}}[n]=\frac{B{_k^{\mathrm{off}}}[n]h_k[n]}{N_0\ln2}$, $\rho_{\mathrm{U},k}^{\mathrm{off}}=\frac{B{_{\mathrm{U},k}
^{\mathrm{off}}}[n]h_{\mathrm{AP}}[n]}{N_0\ln2}$,
$\widehat{\lambda}_{k,n}^*=\sum_{i{\rm{ = }}n+1}^{N-1}\lambda{_{k,i}^*}$ and $\widetilde{\lambda}_{k,n}^*=\sum_{i{\rm{ = }}n}^{N-1}\lambda{_{k,i}^*}$.
Here, $\lambda{_{k,n}^*}\geq0$  for $k\in\mathcal{K}, n\in\mathcal{N}$ are the optimal Lagrange multipliers (dual variables) associated with the inequality constraints \eqref{eq:WSECM1_1}  in problem (P1.1) (or P1), while  $\eta_k^*$ and $\beta^*_k$  are respectively the optimal Lagrange multipliers  associated with the equality constraints \eqref{eq:WSECM1_4} and \eqref{eq:WSECM1_3} for $k\in\mathcal{K}$.
\end{theorem}


It is necessary to obtain the optimal values of the Lagrange multipliers, i.e., $\boldsymbol{\lambda}^*=\{\lambda{_{k,n}^*}\}_{k\in\mathcal{K}, n\in\mathcal{N}}$, $\boldsymbol{\eta}^*=\{\eta{_{k}^*}\}_{k\in\mathcal{K}}$ and  $\boldsymbol{\beta}^*=\{\beta{_{k}^*}\}_{k\in\mathcal{K}}$  since they play important roles in determining the optimal task allocation $\mathbf{l}^*$ according to Theorem \ref{theorem_1}. In this paper, we adopt a subgradient-based algorithm to obtain the optimal dual variables in  $\boldsymbol{\lambda}^*$  related to the inequality constraints \eqref{eq:WSECM1_1} \cite{B_D.Bertsekas89Parallel}.
Besides,  the optimal dual variables in $\boldsymbol{\eta}^*$ and  $\boldsymbol{\beta}^*$  related to the equality constraints \eqref{eq:WSECM1_4} and \eqref{eq:WSECM1_3} can be obtained by bi-section search method.
Note that  the  convergence can be guaranteed according to \cite{B_Boyd04Convex}.

\vspace{-3mm}
\subsection{Bandwidth Allocation}\label{bandwidth_allocation}
Here, another sub-problem of (P1), denoted as the bandwidth allocation problem (P1.2) is considered to optimize $\mathbf{B}$ with the same given UAV's trajectory $\mathbf{u}$  and the optimized computation task allocation parameters in $\mathbf{l}$.
The bandwidth allocation problem (P1.2) is expressed as
\begin{subeqnarray}\label{eq:BA}
 ({\rm P1.2}):
 \underset{\mathbf{B}}{\min} &&\hspace{-4mm}
  \sum\limits_{n{\rm{ = }}1}^N\sum\limits_{k{\rm{ = }}1}^K
  \left(E{_k^{\mathrm{off}}}[n]+E{_{\mathrm{U},k}^{\mathrm{off}}}[n]\right) \quad\quad \slabel{eq:BA_0}\\
\mathrm{s.t.}
&&\hspace{-4mm} \eqref{eq:WSECM1_12},~ \eqref{eq:WSECM1_13}, ~\eqref{eq:WSECM1_15}. \slabel{eq:BA_1} 
\end{subeqnarray}
It can be easily proved that problem (P1.2) is convex with  convex objective function and constraints. To gain more insights on the structure of the optimal solution, we again leverage the Lagrange method \cite{B_Boyd04Convex} to solve this problem, and the optimal solution to problem (P1.2) is given in the following theorem.

\begin{theorem}\label{theorem_2}
The optimal solution of problem (P1.2)  related to UE $k\in\mathcal{K}$ is given by
{\small{
\begin{align}
\hspace{-2mm}B{_k^{\mathrm{off}*}}[n]&=\left\{
\begin{aligned}
&\frac{\frac{\ln2}{2}l_k[n]}{\delta W_0\Big[\frac{\ln2}{2}
\big( \frac{\nu{_{k,n}^*}h_k[n]l_k^{\mathrm{off}}[n]}{\delta^2N_0\ln2}\big)^{\frac{1}{2}}\Big]},~n\in\mathcal{N}_1, \\
&0, \quad n=N-1~\mathrm{or}~N, \\
\end{aligned}\right. \label{B_k_off}\\
\hspace{-2mm}B{_{\mathrm{U},k}^{\mathrm{off}*}}[n]&=\left\{
\begin{aligned}
&\frac{\frac{\ln2}{2}l{_{\mathrm{U},k}^{\mathrm{off}}}[n]}{\delta W_0
\Big[\frac{\ln2}{2}\big(\frac{\nu{_{k,n}^*}h_{\mathrm{AP}}[n]l{_{\mathrm{U},k}^{\mathrm{off}}}[n]}{\delta^2N_0\ln2}\big)
^{\frac{1}{2}}\Big]},~n\in\mathcal{N}_2, \\
&0, \quad n=1~\mathrm{or}~N, \\
\end{aligned}\right. \label{B_Uk_off}
\end{align}
}}
\hspace{-1mm}where $\{\nu{_{k,n}^*}\}_{k\in\mathcal{K}, n\in\mathcal{N}}$  are the optimal Lagrange multipliers associated with the equality constraints in \eqref{eq:WSECM1_12} of problem (P1.2), and  $W_0(x)$ is the principal branch of the Lambert $W$ function defined as the solution of $W_0(x)e^{W_0(x)}=x$ \cite{J_Corless1996LambertW}.
\end{theorem}

\begin{proof}
See Appendix A.
\end{proof}

The optimal Lagrange multipliers $\{\nu{_{k,n}^*}\}$  in Theorem \ref{theorem_2}  should make the  the optimal bandwidth allocation parameters satisfy the equality constraints in \eqref{eq:WSECM1_12}. In fact, $\{\nu{_{k,n}^*}\}$ can be obtained effectively with the bi-section search when the bandwidth is not exclusively occupied, i.e., both $l_k^{\mathrm{off}}[n]$ and  $l{_{\mathrm{U},k}^{\mathrm{off}}}[n]$ are positive,  since $\{B{_k^{\mathrm{off}*}}[n]\}_{n\in\mathcal{N}_1}$ and $\{B{_{\mathrm{U},k}^{\mathrm{off}*}}[n]\}_{n\in\mathcal{N}_2}$  are all monotonically decreasing functions with respect to (w.r.t.) $\{\nu{_{k,n}^*}\}$ according to the property of the $W_0$ function.

\vspace{-2mm}
\subsection{UAV Trajectory Design}\label{UAV_trajectory_design}
Here, the sub-problem for designing the UAV's trajectory $\mathbf{u}$ is considered, which we refer to it as the UAV trajectory design problem (P1.3), by assuming that the computation resource scheduling  $\mathbf{z}$  and bandwidth allocation $\mathbf{B}$ are given as fixed with the previously optimized values. Hence, the  UAV trajectory design problem (P1.3) can be rewritten as
\begin{subeqnarray}\label{eq:UAV}
 ({\rm P1.3}):
 \underset{\mathbf{u}}{\min} &&\hspace{-4mm}
  \sum\limits_{n{\rm{ = }}1}^N\left(E{_{\mathrm{U}}^{(3)}}[n]+\sum\limits_{k{\rm{ = }}1}^KE{_k^{\mathrm{off}}}[n]\right)\quad\quad \slabel{eq:UAV_0}\\
\mathrm{s.t.}
&&\hspace{-4mm} \eqref{eq:WSECM1_9},~\eqref{eq:WSECM1_10}, \slabel{eq:UAV_1} 
\end{subeqnarray}
where $E{_{\mathrm{U}}^{(3)}}[n]=E{_\mathrm{U}^{\mathrm{fly}}}[n]+
\sum\limits_{k{\rm{ = }}1}^K E{_{\mathrm{U},k}^{\mathrm{off}}}[n]$.
Note that the $E{_\mathrm{U}^{\mathrm{fly}}}[n]$ defined in \eqref{eq:E_fly}  is not a convex function of $\mathbf{u}$. In order to address this issue, we define an upper bound of $E{_\mathrm{U}^{\mathrm{fly}}}[n]$ as follows
\begin{align}\label{eq:E_fly_up}
\widetilde{E}{_\mathrm{U}^{\mathrm{fly}}}[n]=\tau\bigg(\theta_1\|\mathbf{v}[n]\|^3+\frac{\theta_2}{\widetilde{v}[n]}\bigg), \ n\in \mathcal{N},
\end{align}
by introducing a variable $\widetilde{v}[n]$ and a constraint $\|\mathbf{v}[n]\|\geq \widetilde{v}[n]$, which is equivalent to $\|\mathbf{u}[n]-\mathbf{u}[n-1]\|^2\geq \widetilde{v}^2[n]\tau^2$.
This constraint  is still non-convex, and we leverage the successive convex approximation (SCA) technique to solve this issue. 
The left hand side of the constraint is convex versus $\mathbf{u}$ and can be  approximated as its linear lower bound by using the first-order Taylor expansion at a local point $\mathbf{u}_i$, where $i=1,2,\dots$ denotes the iteration index of the SCA method. Hence, the additional constraint can be approximated as a convex one as follows
\begin{align}\label{eq:speed_n_app}
&\widetilde{v}^2[n]\tau^2-2(\mathbf{u}_i[n]-\mathbf{u}_i[n-1])^T(\mathbf{u}[n]-\mathbf{u}[n-1]) \\ \nonumber
&\leq \|\mathbf{u}_i[n]-\mathbf{u}_i[n-1]\|^2, \ n\in \mathcal{N}.
\end{align}
The approximated problem of (P1.3) with $\{\widetilde{E}{_\mathrm{U}^{\mathrm{fly}}}[n]\}$, $\{\widetilde{v}[n]\}$ and the constraint \eqref{eq:speed_n_app} is convex w.r.t.~$\mathbf{u}$ and $\{\widetilde{v}[n]\}$, and we can resort to the software CVX \cite{M_Grant08CVX} to solve it.

\vspace{-2mm}
{\small{
\setlength{\tabcolsep}{0.3 pt}\begin{table}[thb]
\centering
\caption{Simulation Parameters}\label{table1}
\begin{tabular}{|l|l|l|}
\hline
~\textbf{Parameter }&~{\textbf{Symbol}} &~{\textbf{Value}} \\
\hline
~System bandwidth\quad~   &~$B$ \quad\quad &~20 MHz \quad\quad\quad\quad \\
\hline
~Task completion time  &~$T$ &~10 seconds\\
\hline
~Number of time slots &~$N$ &~50  \\
\hline
~Number of ground UEs &~$K$ &~4  \\
\hline
~The channel power gain& ~$h_0$ &~$-30\mathrm{dB}$ \\
\hline
~The noise power  &~$N_0$ &~$-60$dBm \\ 
\hline
~Altitude of the UAV&~$H$ &~10 m \\
\hline
~UAV's maximum speed&~$V_{\mathrm{max}}$ &~10 m/s \\
\hline
~UAV's propulsion parameters&~($\theta_1,\theta_2$) &~(0.00614,15.976)\\
\hline
~The switched capacitance &~$\kappa_\mathrm{U}$, $\kappa_k (k\in\mathcal{K})$ &~$10^{-28}$ \\
\hline
~UEs' task-input data size &~$I_k~(k\in\mathcal{K})$ &~400 Mbits \\
\hline
~Required CPU cycles per bit &~$C_k~(k\in\mathcal{K})$ &~1000 cycles/bit \\
\hline
~The tolerant thresholds &~$\epsilon_1$ and $\epsilon$ &~$10^{-4}$  \\
\hline
\end{tabular}
\end{table}
}}

\section{Simulation Results}\label{sec:simulation}

In this section, simulation results are presented to evaluate the effectiveness of the proposed algorithm in comparison with some practical baselines.
The basic simulation parameters are listed in Table~\ref{table1} unless specified otherwise. Besides, the initial and final positions of the UAV  are set as $\mathbf{u}_\mathrm{I}=(-5,-5)$, $\mathbf{u}_\mathrm{F}=(5,-5)$, and the horizontal positions of the UEs are set as $[\mathbf{v}_1$, $\mathbf{v}_2$, $\mathbf{v}_3$, $\mathbf{v}_4]=[(5,5), (-5,5), (-5,-5), (-5,5)]$.

\subsection{Trajectory of the UAV}\label{sec:Trajectory}
\begin{figure}[tbp]
\centering
\includegraphics[scale=0.4]{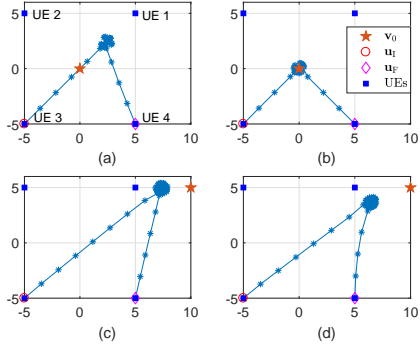}
\caption{The trajectories of the UAV in the situations with \mbox{different} location of the AP and task size allocation of the UEs:   $\mathbf{v}_0=(0,0)$ for (a) and (b), $\mathbf{v}_0=(10,5)$ for (c) and (d);  $[I_1,I_2,I_3,I_4]=[6,2,4,2]\times10^2$Mbits for (a) and (c), $[I_1,I_2,I_3,I_4]=[6,4,6,2]\times10^2$Mbits for (b) and (d).}
\label{Trajectory}
\end{figure}

In this subsection, numerical results for the trajectory of the UAV are given to shed light on the effects of  the task sizes of UEs
($[I_1,I_2,I_3,I_4]$) and the relative location of the AP ($\mathbf{v}_0$), as depicted in Fig. \ref{Trajectory}.
For the scenario of $\mathbf{v}_0=(0,0)$, the AP is surrounded by the UEs and at the center of the UEs' distributed area. We can observe that the UAV tends to fly close to the UEs with large task sizes and tries to be not too far away from the AP when the total task sizes of UEs are moderate as the results in cases (a). When the total task size becomes larger and the distribution of UEs' task sizes becomes  more average, the UAV tends to fly close to the AP as the result in case (b).
These two cases  indicate that for the scenario where the AP is located at the center of UEs' distributed area, the distribution of the UEs' task sizes plays an important role on the UAV's trajectory, while the effect of the AP's location will become more dominant when the UEs' total task size becomes larger, which coincides with the intuition that more task-input data will be offloaded to the AP in this situation so as to reduce the TEC by making use of the super computing resources at the AP. For the scenario of $\mathbf{v}_0=(10,5)$, the AP is located outside the distributed area of the UEs and its average distance to the UEs is relatively larger than the above scenario. 
In this situation, the effects of AP's location on the trajectories are more prominent, where the comparison between (a) and (c), (b) and (d), can properly explain this.

The reason behind these results in Fig.~\ref{Trajectory} is that there exists a tradeoff between the distribution of UEs' task sizes and the relative location of the AP to the UEs. In other words, getting close to the UEs with large task sizes can reduce UEs' offloading energy consumption, while being closer to the AP will reduce the UAV's offloading energy consumption, and thus the UAV has to find a balance between these two factors meanwhile taking its own flying energy consumption into consideration, so as to minimize the TEC through optimizing its flying trajectory.
\vspace{-3mm}
\subsection{Performance Improvement}\label{sec:Performance}
Here, we focus on the performance improvement of the proposed algorithm in comparison with some baseline schemes in the scenario of $\mathbf{v}_0=(0,0)$. The baseline schemes include the the ``Local Computing" scheme where the UEs rely on their own computing resources without task offloading;
the ``Direct Trajectory" scheme where the UAV flies from its initial location to the final location directly with an average speed;
the ``Offloading Only" scheme where the UEs just rely on task offloading to the UAV and the AP for computing without any local computing; the ``Equal Bandwidth" scheme where the whole bandwidth are equally divided for UEs and the UAV's offloading without bandwidth optimization. 

\begin{figure}[tbp]
\centering
\includegraphics[scale=0.4]{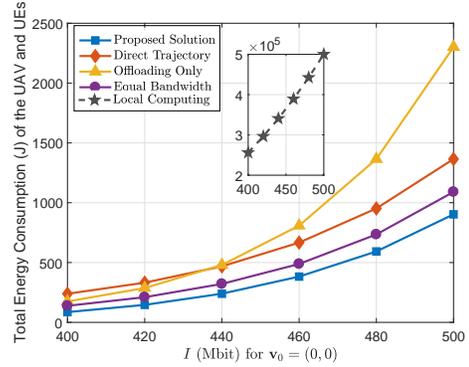}
\caption{The TEC of the UAV and UEs versus the uniform task size: $I=I_k$ for $k\in\mathcal{K}$.}
\label{Energy_I}
\end{figure}

Fig.~\ref{Energy_I} shows the TEC results versus the uniform task size $I=I_k$ for $k\in\mathcal{K}$. All the curves  increase with $I$ as expected since more energy will be consumed by completing tasks with more input data. It can be seen that great performance improvement can be achieved by leveraging the proposed solution in comparison with all the baseline  schemes.
Specifically, the TEC of the ``Proposed Solution" are almost one thousandth of that for the ``Local Computing" scheme, presenting the tremendous benefits the UEs obtained by deploying the UAV as an assistant for computing and  relaying.
In addition, the TEC of the proposed solution are almost third less than those of the other offloading schemes for $I<460$ Mbits.
Moreover, the gaps between the proposed solution and the baseline schemes become larger as $I$ increases, and  the TEC of the proposed solution is  less than  half of the ``Offloading only" scheme when $I=500$ Mbits. All these results verify that the effectiveness of the proposed optimization on task allocation, bandwidth allocation and UAV's trajectory, which make the proposed algorithm superior in dealing with the computation-intensive tasks.


\begin{figure}[tbp]
\centering
\includegraphics[scale=0.4]{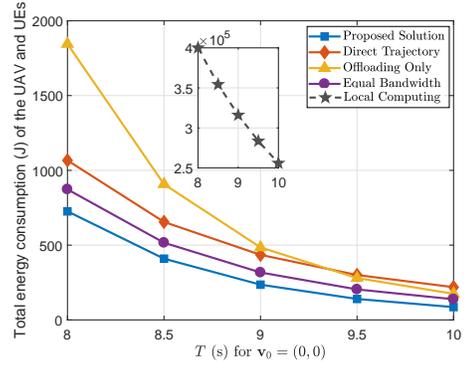}
\caption{The TEC  of the UAV and UEs versus the total task completion time: $T$ (s).}
\label{Energy_T}
\end{figure}

In Fig.~\ref{Energy_T}, the TEC  w.r.t.~the total task completion time $T$ is depicted.
We can see that the TECs of all the schemes decrease with $T$, coinciding with the intuition that a tradeoff exists between the energy consumption and time consumption for completing the same tasks, and the energy consumption will decrease when the consumed time increases.  It is notable that the proposed solution performs better than the four baseline schemes, and the performance improvement is even more prominent with strict time restriction (small $T$), which further confirms that the proposed algorithm is good at dealing with the latency-critical computation tasks and can achieve a better energy-delay tradeoff. 


\section{Conclusion}\label{sec:conclusion}
The UAV-assisted MEC architecture is studied in this paper, where the UAV acts as an MEC server and a relay to help the UEs compute their tasks or further offload their tasks to the AP for computing. The problem of minimizing the TEC of the UAV and the UEs  is solved  by an alternating algorithm, where  the task allocation, the bandwidth allocation, and the UAV's trajectory are iteratively optimized under some practical constraints. The simulation results have confirmed that
significant performance improvement can be achieved by  the proposed algorithm over the baseline schemes especially in handling the computation-intensive and latency-critical tasks.

\section*{Acknowledgment}\label{sec:acknowledgment}
This work is supported  by the U.K. EPSRC under Grants EP/K015893/1.

\section*{Appendix A: Proof of {Theorem}~\ref{theorem_2}}
\label{App:theo_2}
\renewcommand{\theequation}{A.\arabic{equation}}
\setcounter{equation}{0}

The partial Lagrange function of (P1.2) is defined as
{\small{
\begin{align}\label{eq:L_P1_2}
  \mathcal{L}^{(2)}(\mathbf{B},\boldsymbol{\nu})= &
  \sum\limits_{k{\rm{ = }}1}^K \sum\limits_{n{\rm{ = }}1}^N \bigg(E{_k^{\mathrm{off}}}[n]+
  E{_{\mathrm{U},k}^{\mathrm{off}}}[n]\bigg)+ \nonumber \\
  &\sum\limits_{k{\rm{ = }}1}^K\sum\limits_{n{\rm{ = }}1}^N\nu_{k,n}\big(B-B{_k^{\mathrm{off}}}[n]-B{_{\mathrm{U},k}^{\mathrm{off}}}[n]\big),
\end{align}
}}
\hspace{-1.3mm}where $\boldsymbol{\nu}=\{\nu_{k,n}\}_{k\in\mathcal{K},n\in\mathcal{N}}$.
The Lagrangian dual function of problem (P1.2) can be presented as
{\small{
\begin{align}\label{Dual_fun_P1_2}
d^{(2)}(\boldsymbol{\nu})=~\underset{\mathbf{B}}{\min}~ \mathcal{L}^{(2)}(\mathbf{B},\boldsymbol{\nu}),~
~\mathrm{s.t.} ~~ \eqref{eq:WSECM1_13},~\eqref{eq:WSECM1_15}.
\end{align}
}}
\hspace{-1.2mm}Hence, the optimal solution of $\mathbf{B}$ with optimal dual variables $\boldsymbol{\nu}^*$ can be obtained by solving  \eqref{Dual_fun_P1_2}.  It is easy to note that the expressions of $E{_k^{\mathrm{off}}}[n]$ and $E{_{\mathrm{U},k}^{\mathrm{off}}}[n]$   have similar structures w.r.t.~$B{_k^{\mathrm{off}}}[n]$ and $B{_{\mathrm{U},k}^{\mathrm{off}}}[n]$, and thus the optimal solution of $B{_k^{\mathrm{off}}}[n]$ and $B{_{\mathrm{U},k}^{\mathrm{off}}}[n]$  should have similar structures according to problem \eqref{Dual_fun_P1_2}.
Next, we will take $B{_k^{\mathrm{off}}}[n]$  as an example to obtain its closed-form optimal solution versus $\nu_{k,n}^*$ for $k\in\mathcal{K}, n\in\mathcal{N}$.
Applying the KKT conditions  \cite{B_Boyd04Convex} leads to the following necessary and sufficient condition of  $B{_k^{\mathrm{off}*}}[n]$:
{\small{
\begin{align}\label{L_B_k_off}
\hspace{-2mm}\frac{\partial\mathcal{L}^{(2)}(\mathbf{B},\boldsymbol{\nu})}{\partial B{_k^{\mathrm{off}*}}[n]}=\nu_{k,n}^*-\frac{l_k^{\mathrm{off}}[n]N_0\ln2}{(B{_k^{\mathrm{off}*}}[n])^2h_k[n]}
2^{\frac{l_k^{\mathrm{off}}[n]}{B{_k^{\mathrm{off}*}}[n]\delta}}=0,
\end{align}
}}
\hspace{-1.2mm}where the optimal dual variable $\nu_{k,n}^*$ should make sure that the equality constraint  $B{_k^{\mathrm{off}*}}[n]+B{_{\mathrm{U},k}^{\mathrm{off}*}}[n]=B$ is satisfied.
It is not easy to obtain the closed-form solution of $B{_k^{\mathrm{off}*}}[n]$ through \eqref{L_B_k_off} directly.
By defining $\xi=\frac{l_k^{\mathrm{off}}[n]}{B{_k^{\mathrm{off}*}}[n]\delta}$, the equation in \eqref{L_B_k_off} can be re-expressed as $\xi^22^\xi=\frac{\nu_{k,n}^*h_k[n]l_k^{\mathrm{off}}[n]}{\delta^2N_0\ln2}\triangleq\Gamma$,
then by applying the natural logarithm at the both sides leads to $\ln\xi+\frac{\ln2}{2}\xi=\ln\Gamma^{\frac{1}{2}}$, which can be equivalently transformed into $\frac{\ln2}{2}\xi e^{\frac{\ln2}{2}\xi}=\frac{\ln2}{2}\Gamma^{\frac{1}{2}}$ by applying the exponential operation,
where $e$ is the base of the natural logarithm. According to the definition and property of Lambert function \cite{J_Corless1996LambertW}, we have $\frac{\ln2}{2}\xi=W_0(\frac{\ln2}{2}\Gamma^{\frac{1}{2}})$, and finally we can express $B{_k^{\mathrm{off}*}}[n]$ as
\begin{align}\label{express_B_k_off}
B{_k^{\mathrm{off}*}}[n]=\frac{\frac{\ln2}{2}l{_{k}^{\mathrm{off}}}[n]}{\delta W_0\big[\frac{\ln2}{2}(\frac{\nu_{k,n}^*h_k[n]l_k^{\mathrm{off}}[n]}{\delta^2N_0\ln2})^{\frac{1}{2}}\big]},~n\in\mathcal{N}_1.
\end{align}
Integrating with the cases $B{_k^{\mathrm{off}*}}[N]=0$, the complete solution of $B{_k^{\mathrm{off}^*}}[n]$ in \eqref{B_k_off} can be obtained. The solution of $B{_{\mathrm{U},k}^{\mathrm{off}*}}[n]$  in \eqref{B_Uk_off}  can be obtained in a similar way.


\bibliographystyle{IEEEtran}
\bibliography{UAV_MEC_Relay}

\end{document}